\documentclass[pra,amsmath]{revtex4-1}
\usepackage{pifont}
\usepackage{titlesec}
\usepackage{epsfig,dsfont,amssymb,amsmath,amsthm,amsfonts,amsbsy,mathrsfs,amscd}

\usepackage{blkarray}

\newtheorem{theorem}{Theorem}

\newtheorem{example}{Example}

\newtheorem{corollary}{Corollary}

\input amssym.def

\begin{document}

\title{On Ergotropic Gap of Tripartite Separable Systems}
\smallskip
\author{Ya-Juan Wu$^1$}
\author{Shao-Ming Fei$^{2,3}$}
\author{Zhi-Xi Wang$^2$}
\thanks{Corresponding author: wangzhx@cnu.edu.cn}
\author{Ke Wu$^2$}
\thanks{Corresponding author: wuke@cnu.edu.cn}
\affiliation{
{\footnotesize $^1$School of Mathematics, Zhengzhou  University of Aeronautics, Zhengzhou 450046, China}\\
{\footnotesize $^2$School of Mathematical Sciences, Capital Normal University, Beijing 100048, China}\\
{\footnotesize $^3$Max-Planck-Institute for Mathematics in the Sciences, 04103 Leipzig, Germany}
}
\begin{abstract}
Extracting work from quantum system is one of the important areas in quantum thermodynamics. As a significant thermodynamic quantity, the ergotropy gap characterizes the difference between the global and local maximum extractable works. We derive an analytical upper bound of the ergotropic gap with respect to $d\times d\times d$ tripartite separable states. This bound also provides a necessary criterion for the separability of tripartite states. Detailed examples are presented to illustrate the efficiency of this separability criterion.
\end{abstract}

\maketitle

\section{Introduction}
Information plays important roles in thermodynamics \cite{RL, CH, KM}.
Work extraction is a significant aspect in thermodynamics, while quantum correlations are the basic resources in quantum information processing tasks. Their connections have been extensively studied  \cite{RD, KF, JO, RA, WH, SJ, OCO, KMF, VV, HC, KMF}.

Extracting work from quantum systems has gained renewed interest in quantum thermodynamics \cite{LD, HB}. In \cite{GFJ}, the authors addressed the concept of ergotropy: the maximum work that can be gained from a quantum state with respect to some reference Hamiltonian and a cyclical unitary transformations. The maximum extractable work can be naturally divided into two parts: the local contribution from each subsystems, named local ergotropy, and the global contribution originating from correlations among the subsystems, named global ergotropy \cite{MP}. The work that can be extracted is proportional to the quantum mutual information under the global operations of the subsystems \cite{JO, SJ}.
As the extractable work under cyclic local interaction is strictly less than that obtained under global interaction, the presence of quantum correlations among the subsystems may
result in non-vanishing ergotropic gap for the case of non-degenerate energy subspace \cite{AMA}. This fact brings a new insight in the quantum-to-classical transition in thermodynamics.

As an essential kind of correlations, quantum entanglement plays an important role in quantum teleportation, quantum cryptography, quantum dense coding, quantum secret sharing  and the development of quantum computer \cite{C.H1, C.H2, C.H3,  C.B, R.H, Z.Z, DG, YA}.
Distinguishing quantum entangled states from separable ones is of the significant but difficult problems in the theory of quantum entanglement. Numerous entanglement criteria have been proposed, such as the positive partial transpose (PPT) criterion \cite{AP, MHPH}, realignment criteria \cite{OR, KC, CJZ}, correlation matrix or tensor criteria \cite{JI1, JI2, ASM, JI3}, covariance matrix criteria \cite{OG1, OG2}, entanglement witnesses \cite{MHPH, BMT}, separability criteria via measurements \cite{CSM} and so on. In \cite{MA}, Mir Alimuddin {\it et al.} presented a separability criterion (the operational thermodynamic  criterion) based on the evaluation of ergotropic gap of bipartite systems.

In this paper, we focus on the thermodynamic quantity ergotropic gap, given by the difference between the global ergotropy and local ergotropy associated with the tripartite quantum system.
We present an upper bound on the ergotropic gap for arbitrary dimensional tripartite separable states. The bound presents a sufficient criteria for the entanglement of tripartite states.

\section{Bounds on the ergotropic gap of tripartite separable states}


Consider a $d\times d\times d$ tripartite state $\rho_{ABC}\in\mathcal{D}\{\mathcal{H}_A\otimes\mathcal{H}_B\otimes\mathcal{H}_C\}$, where $\mathcal{H}_{X}(X=A, B, C)$ denotes the Hilbert space corresponding to the  subsystem  $X$  and $\mathcal{D}(\mathfrak{X})$ denotes the set of density operator acting on Hilbert space $\mathfrak{X}$. The subsystem X is governed by the local Hamiltonian $H_{X}=\sum\limits_{j=0}^{d-1}j\,E|j\rangle\langle j|$, where $j\,E$ and $|j\rangle$ are the $jth$ energy eigenvalue and eigenvector of $H_X$, respectively. The total non-interacting global Hamiltonian is $H_{ABC}=H_A\otimes I_B\otimes I_C+I_A\otimes H_B\otimes I_C+I_A\otimes I_B\otimes H_C$, where $I_{X}$ denotes the identity operator acting on the Hilbert space $\mathcal{H}_{X}$. Under a cyclic Hamiltonian process, a  time-dependent unitary operation $U(\tau)=\overrightarrow{exp} \left( -i\hbar \int^{\tau }_{0}  dt\left[ H_{ABC}+V(t)\right]  \right) $ can be applied,  which $\overrightarrow{exp}$  denotes the time-ordered exponential and $V(t)$  denotes a time-dependent interaction among the subsystems. The  work extraction from an isolated tripartite system under such a process is the change  in average energy of the system, with the form $\mathcal{W}_e =\mathrm{Tr}((\rho_{ABC}-U\rho_{ABC}U^\dag)H_{ABC})$. Within this framework,  the maximum work,   extracted from the isolated tripartite state $\rho_{ABC}$ by transforming it to the corresponding passive state $\rho_{ABC}^p$, is defined below in (\ref{1}).

The maximum extractable work, called global ergotropy, is defined by
\begin{equation}\label{1}
\begin{split}
\mathcal{W}_e^g &=\max\limits_{U\in\mathcal{L}(\mathcal{H}_A\otimes\mathcal{H}_B
\otimes\mathcal{H}_C)}\mathrm{Tr}((\rho_{ABC}-U\rho_{ABC}U^\dag)H_{ABC})\\
&=\mathrm{Tr}(\rho_{ABC}H_{ABC})-\min\limits_{U\in \mathcal{L}(\mathcal{H}_A\otimes\mathcal{H}_B\otimes\mathcal{H}_C)}\mathrm{Tr}(U\rho_{ABC}U^\dag H_{ABC})\\
&=\mathrm{Tr}(\rho_{ABC}H_{ABC})-\mathrm{Tr}(\rho_{ABC}^pH_{ABC}),
\end{split}
\end{equation}
where $\mathcal{L}(X')$ denotes the set of all bounded linear operators on the Hilbert space $X'$ and  $\rho_{ABC}^p$ is the passive state  with Hamiltonian $H_{ABC}$ for system $ABC$, from which no work can be extracted,   of the form $\rho_{ABC}^p=\sum_{j} \rho_{j} |j\rangle \langle j|$ with $\rho_{j+1} \leqslant \rho_{j}$. $\rho_{ABC}$ and $\rho_{ABC}^p$  have the same spectrum, and therefore there exists a unitary operator $U$ transforming the former to the latter.

The total achievable work, called local ergotropy, is given by
\begin{equation}\label{W}
\mathcal{W}_e^l =\mathcal{W}_e^A+\mathcal{W}_e^B+\mathcal{W}_e^C,
\end{equation}
where $\mathcal{W}_e^A$, $\mathcal{W}_e^B$ and $\mathcal{W}_e^C$ are the maximum local extractable works from systems $A$, $B$ and $C$, respectively,
\begin{equation}
\begin{split}
\mathcal{W}_e^A &=\mathrm{Tr}(\rho_{ABC}H_A\otimes I_B\otimes I_C)-\min\limits_{U\in \mathcal{L}(\mathcal{H}_A)}\mathrm{Tr}((U_A\otimes I_B\otimes I_C)\rho_{ABC}(U_A\otimes I_B\otimes I_C)^\dag H_A\otimes I_B\otimes I_C)\\
&=\mathrm{Tr}(\rho_{ABC}H_A\otimes I_B\otimes I_C)-\mathrm{Tr}(\rho_A^pH_A),
\end{split}
\end{equation}
and, similarly, $\mathcal{W}_e^B =\mathrm{Tr}(\rho_{ABC}I_A\otimes H_B\otimes I_C)-\mathrm{Tr}(\rho_B^pH_B)$,
$\mathcal{W}_e^C=\mathrm{Tr}(\rho_{ABC}I_A\otimes I_B\otimes H_C)-\mathrm{Tr}(\rho_C^pH_C)$
with $\rho_A^p$, $\rho_B^p$ and $\rho_C^p$ the passive states associated with the subsystems $A$, $B$ and $C$, respectively.
Hence,
\begin{equation}
\mathcal{W}_e^l =\mathrm{Tr}(\rho_{ABC}H_{ABC})-\big\{\mathrm{Tr}(\rho_A^p H_A)+\mathrm{Tr}(\rho_B^p H_B)+\mathrm{Tr}(\rho_C^p H_C)\big\}.
\end{equation}

The difference between the global ergotropy and the local ergotropy is called the ergotropic gap $\Delta_{EG}$,
\begin{equation}\label{EG}
\Delta_{EG}=\mathcal{W}_e^g-\mathcal{W}_e^l=\big\{\mathrm{Tr}(\rho_A^p H_A)+\mathrm{Tr}(\rho_B^p H_B)+\mathrm{Tr}(\rho_C^p H_C)\big\}-\mathrm{Tr}(\rho_{ABC}^pH_{ABC}).
\end{equation}
Indeed, $\Delta_{EG}\geq0$ as global unitary operations are capable of extracting work from subsystems as well as from correlations among the subsystems. Clearly the ergotropic gap depends on various kinds of correlations presented in a tripartite quantum system.
It is generally a challenging problem to compute $\Delta_{EG}$ analytically. In the following
we derive analytic upper bounds of the ergotropic gap.
\begin{theorem}
Consider a $d\times d\times d$ tripartite state $\rho_{ABC}$ with spectrum $\lambda(\rho_{ABC})=\{x_{0}, x_{1}, \cdots, x_{d^3-1}\}$ in nonincreasing order. Let the subsystems  be governed by the same Hamiltonian $H_{A}=H_{B}=H_{C}=\sum\limits_{j=0}^{d-1}j\,E|j\rangle\langle j|$. If $\rho_{ABC}$ is separable, then the ergotropic gap is bounded by
\begin{equation}
\Delta_{EG}\leq \min\Big\{(Y-Z)E,M(d)E\Big\},
\end{equation}
where
$$Y=3\sum\limits_{i=0}^{d-1}ix_i+3(d-1)\sum\limits_{i=d}^{d^3-1}x_i,$$
$$Z=\sum\limits_{i=1}^{d-1}i\sum\limits_{j_i'=0}^{\frac{(i+1)(i+2)}{2}-1}x_{D_i+j_i'}+
\sum\limits_{i=1}^{d-1}(d-1+i)\sum\limits_{k_i'=0}^{\frac{(d+i)(d+i+1)}{2}-\frac{3}{2}i(i+1)-1}x_{D_{d+i-1}-3D_{i-1}+{k_i'}}
+\sum\limits_{i=1}^{d-1}(2d-2+i)\sum\limits_{{l_i'}=0}^{\frac{(d-i)(d-i+1)}{2}-1}x_{D-D_{d-i}+{l_i'}},$$
and
$$D=\frac{(2d-1)2d(2d+1)}{6}-\frac{(d-1)d(d+1)}{3},$$$$D_i=\frac{i(i+1)(i+2)}{6},$$
$$M(d)=\frac{3(d-1)}{2}-\frac{l}{d}[\frac{l^3+2l^2-5l+2}{8}+m+1].$$
The integers $l$ and $m$ are uniquely determined by the constraint $\frac{(l-1)(l+1)(l+2)}{6}+m=d-1$, where $0\leq m\leq \frac{(l+1)(l+2)}{2}$.
\end{theorem}

\begin{proof}
Let the spectra of the reduced sub-states $\rho_A$, $\rho_B$ and $\rho_C$ be
$\lambda(\rho_A)=\{p_0, p_1, \cdots, p_{d-1}\}$, $\lambda(\rho_B)=\{q_0, q_1, \cdots, q_{d-1}\}$ and $\lambda(\rho_C)=\{r_0, r_1, \cdots, r_{d-1}\}$, respectively. Without loss of generality, we assume that the spectra are arranged in nonincreasing order.

i) {\it Proof of $\Delta_{EG}\leq (Y-Z)E$}

From \eqref{EG}, the ergotropy gap of the system can be written in the following form:
\begin{equation}\label{th1}
\Delta_{EG}=\sum\limits_{i=0}^{d-1}i(p_i+q_i+r_i)E-\mathrm{Tr}(\rho_{ABC}^pH_{ABC}),
\end{equation}
where the first three terms are the local passive state energies of subsystems $A$, $B$ and $C$, respectively, and the last term is the global passive state energy.

A state $\rho$ is said to be majorized by a state $\sigma$,
$\lambda(\rho)\prec\lambda(\sigma)$, if
$\sum\limits_{i=1}^kp_i^\downarrow \leq \sum\limits_{i=1}^kq_i^\downarrow$
for $1\leq k\leq n-1$ and $\sum\limits_{i=1}^np_i^\downarrow = \sum\limits_{i=1}^nq_i^\downarrow$,
where $\lambda(\rho)\equiv\{p_i^\downarrow\}$ and $\lambda(\sigma)\equiv \{q_i^\downarrow\}$
are the spectra of $\rho$ and $\sigma$, respectively, arranged in nonincreasing order.
By convention one appends zeros to make the two vectors $\lambda(\rho)$ and $\lambda(\sigma)$ have the same dimensions.
From the Nielsen-Kempe separability criterion \cite{MAN},
if $\rho_{ABC}$ is separable, one has
$\lambda(\rho_A) \succ \lambda(\rho_{ABC})  \quad\wedge \quad\lambda(\rho_B) \succ \lambda(\rho_{ABC})  \quad\wedge\quad  \lambda(\rho_C)\succ\lambda(\rho_{ABC})$.
 Namely, e.g. for $\rho_A$,
\begin{equation} \label{xiu1}
  \begin{split}
             p_0\geq x_0, ~~
             p_0+p_1\geq x_0+x_1,~~
             \cdots,~~
             p_0+\cdots+p_i\geq x_0+\cdots+x_i ,~~
             \cdots,~~\\
           \sum\limits_{i=0}^{d-2}p_i\geq \sum\limits_{i=0}^{d-2}x_i,~~
             \cdots, ~~
             \sum\limits_{i=0}^{d^3-2}p_i\geq \sum\limits_{i=0}^{d^3-2}x_i,~~
             \sum\limits_{i=0}^{d-1}p_i=\sum\limits_{i=0}^{d^3-1}p_i= \sum\limits_{i=0}^{d^3-1}x_i,~~
             \end{split}
          \end{equation}
          by substituting  $p_i$ for $q_i$ and $r_i$, similar results are obtained  for $\rho_B$ and $\rho_C$, respectively.

    Subtracting the last term by  $\mathrm{1st~ term}, \mathrm{2nd ~term}, \cdots, (d^3-1)$-th term in \eqref{xiu1},  respectively, we get
\begin{equation}\label{xiu8}
 \begin{split}
             \sum\limits_{i=1}^{d-1}p_i \leq \sum\limits_{i=1}^{d^3-1}x_i,~~
             \sum\limits_{i=2}^{d-1}p_i \leq \sum\limits_{i=2}^{d^3-1}x_i,~~
             \cdots,~~\\
             \sum\limits_{i=j}^{d-1}p_i\leq \sum\limits_{i=j}^{d^3-1}x_i ,~~
             \cdots,~~
             p_{d-1}\leq \sum\limits_{{i=}d-1}^{d^3-1}x_i.~~
              \end{split}
 \end{equation}
 Similar inequalities can be obtained  for $\rho_B$ and $\rho_C$, respectively.

Summing over the above inequalities for $\rho_A$, $\rho_B$ and $\rho_C$ , we obtain
\begin{equation}\label{th2}
\sum\limits_{i=0}^{d-1}i(p_i+q_i{+}r_i)\leq3\sum\limits_{i=0}^{d-1}ix_i+3(d-1)\sum\limits_{i=d}^{d^3-1}x_i.
\end{equation}
From \eqref{th2} into \eqref{th1}, we get the bound of the ergotropic gap,
\begin{equation}\label{th3}
\Delta_{EG}\leq 3\sum\limits_{i=0}^{d-1}ix_iE+3(d-1)\sum\limits_{i=d}^{d^3-1}x_iE-\mathrm{Tr}(\rho_{ABC}^pH_{ABC})\triangleq (Y-Z)E,
\end{equation}
where
$Y\equiv 3\sum\limits_{i=0}^{d-1}ix_i+3(d-1)\sum\limits_{i=d}^{d^3-1}x_i$
and $ZE=\mathrm{Tr}(\rho_{ABC}^pH_{ABC})$.

To evaluate $ZE=\mathrm{Tr}(\rho_{ABC}^pH_{ABC})=\min\limits_{U\in \mathcal{L}(\mathcal{H}_A\otimes\mathcal{H}_B\otimes\mathcal{H}_C)}\mathrm{Tr}(U\rho_{ABC}U^\dag H_{ABC})$,  note that the total non-interacting global Hamiltonian
\begin{equation}
\begin{split}
H_{ABC} &=H_A\otimes I_B\otimes I_C+I_A\otimes H_B\otimes I_C+I_A\otimes I_B\otimes H_C\\
&= diag\{0,\cdots,d-1\}E\otimes diag\{1,\cdots,1\}\otimes diag\{1,\cdots,1\}+diag\{1,\cdots,1\}\otimes diag\{0,\cdots,d-1\}E\\
&\otimes diag\{1,\cdots,1\}+diag\{1,\cdots,1\}\otimes diag\{1,\cdots,1\}\otimes diag\{0,\cdots,d-1\}E \\
&= diag\{0,\cdots,d-1,\cdots,d-1,\cdots,2d-2,\cdots,d-1,\cdots,2d-2,\cdots,2d-2,\cdots,3d-3\}E.
\end{split}
\end{equation}

To obtain the minimum value, we need to designate the corresponding spectrum of the passive state $\rho_{ABC}^p$:

Energy 0:     $x_0\rightarrow|000\rangle.$

Energy 1:     $x_1\rightarrow|100\rangle;~   x_2\rightarrow|010\rangle;~ x_3\rightarrow|001\rangle.$

Energy 2:      $x_4\rightarrow|200\rangle;~   x_5\rightarrow|110\rangle;~ x_6\rightarrow|020\rangle;$~
               $x_7\rightarrow|011\rangle;~  x_8\rightarrow|002\rangle;~ x_9\rightarrow|101\rangle.$

Energy 3:      $x_{10}\rightarrow|300\rangle;~   x_{11}\rightarrow|210\rangle;~ x_{12}\rightarrow|120\rangle;$~
               $x_{13}\rightarrow|030\rangle;~  x_{14}\rightarrow|021\rangle;~ x_{15}\rightarrow|012\rangle;$~
               $x_{16}\rightarrow|003\rangle;~  x_{17}\rightarrow|102\rangle;~ x_{18}\rightarrow|201\rangle;~x_{19}\rightarrow|111\rangle.$

    $\cdots$.

Energy $i$:      $x_{\frac{i(i+1)(i+2)}{6}}\triangleq x_{D_i}\rightarrow|i00\rangle;~   x_{D_i+1}\rightarrow|(i-1)10\rangle;~\cdots;~ x_{D_i+i}\rightarrow|0i0\rangle,~$
             $x_{D_i+i+1}\rightarrow|0(i-1)1\rangle;~ \cdots;~ $\\
             $x_{D_i+2i}\rightarrow|00i\rangle;~\cdots;~ x_{D_i+j_i'}\rightarrow|j'_{i1}j'_{i2}j'_{i3}\rangle;~ \cdots;~x_{D_{i+1}-1},$~
    where $0\leq j_i'\leq\frac{(i+1)(i+2)}{2}-1$ and $j'_{i1}+j'_{i2}+j'_{i3}=i$.

    $\cdots$.

Energy $d-1$:      $x_{D_{d-1}}\rightarrow|(d-1)00\rangle;~ x_{D_{d-1}+1}\rightarrow|(d-2)10\rangle;~\cdots; x_{D_{d-1}+j_{d-1}'}\rightarrow|j'_{(d-1)1}j'_{(d-1)2}j'_{(d-1)3}\rangle;~ \cdots;~x_{D_d-1},$\\where $0\leq j_{d-1}'\leq\frac{d(d+1)}{2}-1$ and  $j'_{(d-1)1}+j'_{(d-1)2}+j'_{(d-1)3}=d-1$.

Energy $d$:      $x_{D_d}\rightarrow|(d-1)10\rangle;~ x_{{D_d}+1}\rightarrow|(d-2)20\rangle;~ \cdots;~ x_{D_d+k_1'}\rightarrow|k'_{11}k'_{12}k'_{13}\rangle;~ \cdots;~x_{D_d+\frac{(d+1)(d+2)}{2}-\frac{3}{2}(1\times 2)}= x_{D_{d+1}-3D_1-1},$~
    where $0\leq k_1'\leq \frac{(d+1)(d+2)}{2}-\frac{3}{2}(1\times 2)-1$ and  $k'_{11}+k'_{12}+k'_{13}=d$ with $ k'_{11},k'_{12},k'_{13}\leq d-1.$

  $\cdots$.

Energy $d+i-1$:      $x_{D_{d+i-1}-3D_{i-1}} \rightarrow|(d-1)i0\rangle;~x_{D_{d+i-1}-3D_{i-1}+1} \rightarrow|(d-2)(i+1)0\rangle;~\cdots; x_{D_{d+i-1}-3D_{i-1}+k_i'}\rightarrow|k'_{i1}k'_{i2}k'_{i3}\rangle;~ \cdots;~x_{D_{d+i}-3D_i-1},$~
    where $0\leq k_i'\leq\frac{(d+i)(d+i+1)}{2}-\frac{3}{2}i(i+1)-1$ and  $k'_{i1}+k'_{i2}+k'_{i3}=d+i-1$ with $ k'_{i1},k'_{i2},k'_{i3}\leq d-1.$

     $\cdots$.

Energy $2d-2$:      $x_{D_{2d-2}-3D_{d-2}}\rightarrow|(d-1)(d-1)0\rangle;~ x_{D_{2d-2}-3D_{d-2}+1}\rightarrow|(d-2)d0\rangle;~\cdots; x_{D_{2d-2}-3D_{d-2}+k_{d-1}'}\rightarrow|k'_{(d-1)1}k'_{(d-1)2}k'_{(d-1)3}\rangle;~ \cdots;~x_{D_{2d-1}-3D_{d-1}-1},$~
    where $0\leq k_{d-1}'\leq\frac{(2d-1)2d}{2}-\frac{3}{2}(d-1)\times d -1$ and  $k'_{(d-1)1}+k'_{(d-1)2}+k'_{(d-1)3}=2d-2$ with  $k'_{(d-1)1},k'_{(d-1)2},k'_{(d-1)3}\leq d-1.$

Energy $2d-1$:      $x_{D_{2d-1}-3D_{d-1}}\triangleq x_{D-D_{d-1}} \rightarrow|(d-1)(d-1)1\rangle;~  x_{D-D_{d-1}+1} \rightarrow|(d-2)(d-1)2\rangle;~\cdots; x_{D-D_{d-1}+l_1'}\rightarrow|l'_{11}l'_{12}l'_{13}\rangle;~ \cdots;~ x_{D-D_{d-2}-1},$~
    where $D=\frac{(2d-1)2d(2d+1)}{6}-\frac{(d-1)d(d+1)}{3}$, $0\leq l_1'\leq\frac{(d-1)d}{2}-1$ and $l'_{11}+l'_{12}+l'_{13}=2d-1$ with $l'_{11},l'_{12},l'_{13}\leq d-1$.

     $\cdots$.

Energy $2d+i-2$:      $x_{D-D_{d-i}}\rightarrow|(d-1)(d-1)i\rangle ;~ x_{D-D_{d-i}+1}\rightarrow|(d-2)(d-1)(i+1)\rangle ;~ \cdots; x_{D-D_{d-i}+l_i'}\rightarrow|l'_{i1}l'_{i2}l'_{i3}\rangle;~ \cdots;~ x_{D-D_{d-i-1}-1},$~
    where $0\leq l_i'\leq \frac{(d-i)(d-i+1)}{2}-1$ and  $l'_{i1}+l'_{i2}+l'_{i3}=2d+i-2$ with $l'_{i1},l'_{i2},l'_{i3}\leq d-1$.

 $\cdots$.

Energy $3d-3$:      $x_{d^3-1}\rightarrow|(d-1)(d-1)(d-1)\rangle$.

For convenience, we give a diagram to represent the above specification for the corresponding spectrum of the passive state $\rho_{ABC}^p$. In this diagram,  elements in each row  have the same energy, as follows:
\begin{equation}
\bordermatrix{
0        & x_0                     &             &                 \cr
1        & x_1                     &\cdots       &x_3                  \cr
\vdots   &\vdots                   &             && \cr
i       &x_{D_i}                   &\cdots   &x_{D_i+{j_i'}}    &\cdots&x_{D_{i+1}-1}      \cr
\vdots  &\vdots                    &                                                 \cr
d-1   &x_{D_{d-1}}             &\cdots       &\cdots&\cdots&x_{D_{d}-1} \cr
d     &x_{D_{d}}               &\cdots       &\cdots&\cdots&\cdots&x_{D_{d+1}-3D_1-1}    \cr
\vdots  &\vdots                    &               \cr
d+i-1 &x_{D_{d+i-1}-3D_{i-1}}  &\cdots       &\cdots&x_{D_{d+i-1}-3D_{i-1}+{k_i'}}&\cdots&\cdots&x_{D_{d+i}-3D_i-1}  \cr
\vdots  &\vdots                      &             \cr
2d-2  &x_{D_{2d-2}-3D_{d-2}}  &\cdots      &\cdots&\cdots&\cdots&x_{D_{2d-1}-3D_{d-1}-1}      \cr
2d-1  &x_{D-D_{d-1}}           &\cdots       &\cdots&\cdots&x_{D-D_{d-2}-1}     \cr
\vdots  &\vdots                    &                 \cr
2d+i-2&x_{D-D_{d-i}}              &\cdots  &x_{D-D_{d-i}+{l_i'}}  &\cdots&x_{D-D_{d-i-1}-1}      \cr
\vdots  &\vdots                     &                   \cr
3d-3  &x_{d^3-1}                     &           &       \cr
}.
\end{equation}

Therefore,
\begin{equation}\label{th4}
\begin{split}
ZE&=\mathrm{Tr}(\rho_{ABC}^pH_{ABC})\\&=\sum\limits_{i=1}^{d-1}iE\sum\limits_ {{j_i'}=0}^{\frac{(i+1)(i+2)}{2}-1}x_{D_i+{{j_i'}}}+
\sum\limits_{i=1}^{d-1}(d-1+i)E\sum\limits_{{k_i'}=0}^{\frac{(d+i)(d+i+1)}{2}-\frac{3}{2}i(i+1)-1}x_{D_{d+i-1}-3D_{i-1}+{k_i'}}\\
&+\sum\limits_{i=1}^{d-1}(2d-2+i)E\sum\limits_{{l_i'}=0}^{\frac{(d-i)(d-i+1)}{2}-1}x_{D-D_{d-i}+{l_i'}},
\end{split}
\end{equation}
where the last three terms in \eqref{th4} correspond to the spectra $\{x_0, x_1, \cdots, x_{D_{d}-1}\}$, $\{x_{D_{d}}, \cdots, x_{D_{2d-1}-3D_{d-1}-1}\}$ and $\{x_{D-D_{d-1}}=x_{D_{2d-1}-3D_{d-1}}, \cdots, x_{d^3-1}\}$, respectively.

Substituting \eqref{th4} in \eqref{th3}, we get
\begin{equation}
\begin{split}
\Delta_{EG}&\leq \bigg(3\sum\limits_{i=0}^{d-1}ix_i+3(d-1)\sum\limits_{i=d}^{d^3-1}x_i-\sum\limits_{i=1}^{d-1}i\sum\limits_{{{j_i'}=0}}^{\frac{(i+1)(i+2)}{2}-1}x_{D_i+{{j_i'}}}\\
&-\sum\limits_{i=1}^{d-1}(d-1+i)\sum\limits_{{k_i'}=0}^{\frac{(d+i)(d+i+1)}{2}-\frac{3}{2}i(i+1)-1}x_{D_{d+i-1}-3D_{i-1}+{k_i'}}-\sum\limits_{i=1}^{d-1}(2d-2+i)\sum\limits_{{l_i'}=0}^{\frac{(d-i)(d-i+1)}{2}-1}x_{D-D_{d-i}+{l_i'}}\bigg)E,
\end{split}
\end{equation}
which proves $\Delta_{EG}\leq (Y-Z)E$.

ii) {\it Proof of $\Delta_{EG}\leq M({ d})E$}

Similar to the approach used in \cite{MA}, we rewrite \eqref{th1} as,
\begin{equation}\label{th5}
\begin{split}
\Delta_{EG}&=\sum\limits_{i=0}^{d-1}i(p_i+q_i+r_i)E
-\sum\limits_{i=1}^{l-1}iE\sum\limits_{k'=0}^{\frac{(i+1)(i+2)}{2}-1}r_{D_i+k'}-lE\sum\limits_{k'=0}^m r_{D_l+k'}\\
&+\sum\limits_{i=1}^{l-1}iE\sum\limits_{k'=0}^{\frac{(i+1)(i+2)}{2}-1}r_{D_i+k'}+lE\sum\limits_{k'=0}^m r_{D_l+k'}
-\mathrm{Tr}(\rho_{ABC}^pH_{ABC}),
\end{split}
\end{equation}
where $l$ and $m$ are integers determined uniquely by the constraint
\begin{equation}
D_i+m=\frac{l(l+1)(l+2)}{6}+m=d-1, \quad 0\leq m\leq\frac{(l+1)(l+2)}{2}.
\end{equation}

Replacing $p_i$ with $r_i$ in  \eqref{xiu8},  we obviously have
\begin{equation}\label{2xiu}
\sum\limits_{i=1}^{l-1}iE\sum\limits_{k'=0}^{\frac{(i+1)(i+2)}{2}-1}r_{D_i+k'}+lE\sum\limits_{k'=0}^m r_{D_l+k'}\leq\sum\limits_{i=1}^{l-1}iE\sum\limits_{k'=0}^{\frac{(i+1)(i+2)}{2}-1}x_{D_i+k'}+lE\sum\limits_{k'=0}^m x_{D_l+k'}.
\end{equation}

    Putting the expressions of  \eqref{th4} and \eqref{2xiu} into \eqref{th5}, we get
\begin{equation}\label{th6}
\begin{split}
\Delta_{EG}&\leq\sum\limits_{i=0}^{d-1}i(p_i+q_i+r_i)E
-\sum\limits_{i=1}^{l-1}iE\sum\limits_{k'=0}^{\frac{(i+1)(i+2)}{2}-1}r_{D_i+k'}-lE\sum\limits_{k'=0}^m r_{D_l+k'}+\sum\limits_{i=1}^{l-1}iE\sum\limits_{k'=0}^{\frac{(i+1)(i+2)}{2}-1}x_{D_i+k'}\\
&+lE\sum\limits_{k'=0}^m x_{D_l+k'}
-\bigg[\sum\limits_{i=1}^{d-1}iE\sum\limits_{{j_i'}=0}^{\frac{(i+1)(i+2)}{2}-1}x_{D_i+{j_i'}}+
\sum\limits_{i=1}^{d-1}(d-1+i)E\sum\limits_{{k_i'}=0}^{\frac{(d+i)(d+i+1)}{2}-\frac{3}{2}i(i+1)-1}x_{D_{d+i-1}-3D_{i-1}+{k_i'}}\\
&
+\sum\limits_{i=1}^{d-1}(2d-2+i)E\sum\limits_{{l_i'}=0}^{\frac{(d-i)(d-i+1)}{2}-1}x_{D-D_{d-i}+{l_i'}}\bigg].
\end{split}
\end{equation}

Note that
\begin{equation}\label{2xiuth7}
\sum\limits_{i=0}^{d-1}ir_iE
-\sum\limits_{i=1}^{l-1}iE\sum\limits_{k'=0}^{\frac{(i+1)(i+2)}{2}-1}r_{D_i+k'}-lE\sum\limits_{k'=0}^m r_{D_l+k'}=\sum\limits_{i=1}^{\frac{l(l+1)(l+2)}{6}-1}(i-l')r_iE+\sum\limits_{i=\frac{l(l+1)(l+2)}{6}}^{d-1}(i-l)r_iE,
\end{equation}
where $(l', m')$ are determined by $i=\frac{l'(l'+1)(l'+2)}{6}+m'$, $0\leq m'\leq\frac{(l'+1)(l'+2)}{2}$.
Substituting \eqref{2xiuth7} in \eqref{th6}, we have
\begin{equation}\label{th7}
\Delta_{EG}\leq\sum\limits_{i=0}^{d-1}ip_iE+\sum\limits_{i=0}^{d-1}iq_iE
+\sum\limits_{i=1}^{\frac{l(l+1)(l+2)}{6}-1}(i-l')r_iE
+\sum\limits_{i=\frac{l(l+1)(l+2)}{6}}^{d-1}(i-l)r_iE-\delta E,
\end{equation}
where
\begin{equation}
\begin{split}
\delta &\triangleq\bigg[\sum\limits_{i=1}^{d-1}i\sum\limits_{{j_i'}=0}^{\frac{(i+1)(i+2)}{2}-1}x_{D_i+{j_i'}}+
\sum\limits_{i=1}^{d-1}(d-1+i)\sum\limits_{{k_i'}=0}^{\frac{(d+i)(d+i+1)}{2}-\frac{3}{2}i(i+1)-1}x_{D_{d+i-1}-3D_{i-1}+{k_i'}}\\
&+\sum\limits_{i=1}^{d-1}(2d-2+i)\sum\limits_{{l_i'}=0}^{\frac{(d-i)(d-i+1)}{2}-1}x_{D-D_{d-i}+{l_i'}}\bigg]-\bigg(\sum\limits_{i=1}^{l-1}i\sum\limits_{k'=0}^{\frac{(i+1)(i+2)}{2}-1}x_{D_i+k'}+l\sum\limits_{k'=0}^m x_{D_l+k'}\bigg)\geq0.
\end{split}
\end{equation}

Maximizing the right hand side of \eqref{th7} we have $\min(\delta)=0$,
$p_i=q_i=r_i={1}/{d}$ for $i=1,...,d-1$, and
\begin{equation}
\begin{split}
\Delta_{EG}&\leq M(d)E=\bigg(\frac{d-1}{2}+\frac{d-1}{2}+\frac{d-1}{2}-\frac{1}{d}\sum\limits_{i=1}^{l-1}i\frac{(i+1)(i+2)}{2}-\frac{l(m+1)}{d}\bigg)E\\
&=\bigg(\frac{3(d-1)}{2}-\frac{l}{d}\big(\frac{l^3+2l^2-5l+2}{8}+m+1\big)\bigg)E.
\end{split}
\end{equation}
Altogether we have $\Delta_{EG}\leq \min\Big\{(Y-Z)E,M(d)E\Big\}$.
\end{proof}

From the Theorem, we have the following separability criterion
\begin{corollary}
A tripartite state $\rho_{ABC}$ is entangled if
\begin{equation}\label{en1}
\Delta_{EG}> \min\Big\{(Y-Z)E,M(d)E\Big\}.
\end{equation}
\end{corollary}

In particular, for three-qubit case, the reduced states are just qubit ones. The local qubit systems are governed by the same two-energy levels Hamiltonian $H=E|1\rangle\langle1|$.
For separable three-qubit states $\rho_{ABC}$ with the spectrum $\{x_0, x_1, \cdots, x_7\}$ in nonincreasing order, the ergotropy gap is bounded by
\begin{equation}
\Delta_{EG}\leq \min \Big\{[2(x_1+x_2+x_3)+x_4+x_5+x_6]E, E \Big\}.
\end{equation}

Let us consider several examples.

\begin{example}
The superposition of W and GHZ states \cite{RLA}:
$|\psi\rangle=\sqrt{p}|GHZ\rangle+\sqrt{1-p}|W\rangle$, $0\leq p\leq1$,
where
$|GHZ\rangle=\frac{1}{\sqrt{2}}(|000\rangle+|111\rangle)$
and
$|W\rangle=\frac{1}{\sqrt{3}}(|100\rangle+|010\rangle+|001\rangle)$.
For this state, we have
$\Delta_{EG}=\frac{3-\sqrt{1+4p-5p^2}}{2}E$, $Y-Z=0$ and $M(d)=1$. From (\ref{en1}) we have that $|\psi\rangle$ is entangled for  $0\leq p\leq 1$.

The entanglement criterion for GHZ states \cite{YA} claims that  any separable 3-qubit state $\rho$ satisfies the following inequalities
$$\mid \left< A_{1}\right>_{\rho } \pm \left< A_{2}\right>_{\rho }  +\left< A_{3}\right>_{\rho }  \mid \leq 1, ~~\mid \left< A_{1}\right>_{\rho }  \pm \left< A_{2}\right>_{\rho }  -\left< A_{3}\right>_{\rho }  \mid \leq 1,$$
where $A_1=\sigma_{x} \otimes \sigma_{x} \otimes \sigma_{x} ,$ $A_2=I \otimes \sigma_{z} \otimes \sigma_{z} ,$ $A_3=\sigma_{y} \otimes \sigma_{y} \otimes \sigma_{x} , $ and $\left< A_{i}\right>_{\rho }= Tr(\rho A_i).$
The inequalities may be violated by entangled states.
 It has been shown that $|\psi\rangle$ is entangled for $\frac{2}{5}< p\leq 1$.

 And the entanglement criterion for W states \cite{YA} claims that any separable 3-qubit state $\rho$ satisfies the following inequalities
$$\mid \left< B_{1}\right>_{\rho }  \pm \left< B_{2}\right>_{\rho }  +\left< B_{3}\right>_{\rho}  \mid \leq 1, ~~\mid \left< B_{1}\right>_{\rho }  \pm \left< B_{2}\right>_{\rho }  -\left<B_{3}\right>_{\rho}  \mid \leq 1,$$
where $B_1=I \otimes \sigma_{x} \otimes \sigma_{x} ,$ $B_2=I \otimes \sigma_{y} \otimes \sigma_{y} ,$ $B_3=\sigma_{z} \otimes \sigma_{z} \otimes \sigma_{z} ,$ and $\left< B_{i}\right>_{\rho }= Tr(\rho B_i).$
The inequalities may be violated by entangled states.
 It has been shown that $|\psi\rangle$ can be identified as an entangled state for $0 \leq p < \frac{4}{7}$. It is obvious  that our result is an improvement.

\end{example}

\begin{example}
The GHZ state mixed with a colored noise \cite{ACA, MSP},
$$
\rho=\frac{p}{2}(|000\rangle\langle 000|+|111\rangle\langle111|)+(1-p)|GHZ\rangle\langle GHZ|, ~~~0\leq p\leq1.
$$
For this state, we have
$\Delta_{EG}=\big(\frac{3}{2}-\frac{p}{2}\big)E$, $Y-Z=p$ and $M(d)=1$.
Therefore $\rho_{ABC}$ is entangled for $p<1$, which is the same result obtained in \cite{YA}.
\end{example}

\begin{example}
The GHZ state mixed with the white noise,
$$
\rho=(1-p)\frac{I}{8}+p|GHZ\rangle\langle GHZ|,~~~ 0\leq p\leq1.
$$
For this state, we have
$\Delta_{EG}=\frac{3}{2}pE$, $Y-Z=\frac{9}{8}(1-p)$ and $M(d)=1$.
Hence $\rho_{ABC}$ is entangled for $p>\frac{3}{7}$, as the result obtained in \cite{OG, SS}.
\end{example}

\section{Conclusion}
We have investigated the ergotropic gap, the difference between the global ergotropy and local ergotropy for tripartite systems. The upper bounds of the ergotropic gap have been analytically
derived. The violation of the bound provides a sufficient criterion for entanglement.
The ergotropic gap is tightly related to the work extraction and correlations among the subsystems.
Our results may highlight further studies on ergotropic gaps for multipartite systems the detection of bi-separability and the genuine multipartite entanglement.

\section{Acknowledgement}
This work is supported by the National Natural Science Foundation of China under
Grant Nos. 11871350,  12075159 and 12171044; Beijing Natural Science Foundation(Grant No.
Z190005); Academy for Multidisciplinary Studies, Capital Normal University;  the Academician Innovation Platform of Hainan Province; Shenzhen Institute for Quantum Science and Engineering, Southern University of Science and
Technology(No. SIQSE202001).

\end{document}